\documentclass[journal,11pt,draftcls,onecolumn]{IEEEtran}
\usepackage{amsmath,amssymb,amsbsy,amsthm,geometry}
\usepackage{algorithm} 
\usepackage{algpseudocode} 

\let\epsilon\varepsilon
\let\phi\varphi

\let\epsilon\varepsilon

\newtheorem*{lemma*}{Lemma}
\newtheorem{theorem}{Theorem}
\newtheorem{lemma}{Lemma}

\newtheorem{definition}{Definition}

\usepackage[usenames]{color}

\begin{document}

\title{ Perfect  message authentication codes are robust to small deviations from uniform key distributions.
}

\author{
 \IEEEauthorblockN{Boris Ryabko\\}
 \IEEEauthorblockA{Federal Research Center for Information and Computational Technologies,
\\Novosibirsk State University, 
and  \\Siberian State University of Telecommunications and Informatics, 
\\
Novosibirsk\\}}

\date{}

\maketitle

\begin{center}

    \end{center}

Avstract. 
 We investigate the impact of (possible) deviations of the probability distribution of key values from a uniform distribution for the information-theoretic strong, or perfect,  message authentication code.   We found a simple expression for the decrease in security as a function of the statistical distance between the real key probability distribution and the uniform one.   In a sense, a perfect message authentication code is robust to small deviations from a uniform key distribution. 

\section{Introduction}

Perfect security is an important property of data protection systems, which has attracted the attention of cryptography researchers since C. Shannon described it in his famous paper \cite{sh}, where he also showed that the so-called one-time pad cipher possesses this property. Several perfectly secure cryptographic methods are currently known, which include the  information-theoretically secure    (i,e.        the perfectly secure) message authentication code (MAC), which is popular and well-studied, cf. \cite{bo,wk}.
This MAC, as well as the one-time pad,    use secret keys, i.e., binary words of a certain length, which must obey a uniform distribution.

It is worth noting that MACs  and related cryptographic primitives are under intensive development by many authors \cite{bo,wk,c1,c2,c3,c4}, and there are currently some differences in terminology.
In this article, all definitions are given according to \cite{bo}.

In this paper, we consider the problem when the actual probability distribution of key is different from uniform  one i.e., there is a small difference between the actual key distribution and the uniform distribution.
We found estimates of the reduction in MAC persistence with non-uniformly distributed keys by measuring the deviation using the so-called statistical distance, which is a popular measure of the difference of distributions \cite{bo}. 
The estimate found is, in a sense, a generalization of the case of uniformly distributed keys.

To the best of our knowledge, this problem for MAC has not been solved yet, and the results presented are new.

One informal observation is that the perfect security property of MAC is robust to small deviations from key randomness.

The rest of the paper is as follows. Part 2 contains a general definition of MAC and a description of perfect, or information-theoretically secure,  MAC. The third part is devoted to the description of the so-called statistical distance and its some properties, while the fourth part contains the proof of the main results about MAC security for the case when the key probability distribution is slightly different from the uniform distribution. 

\section{ Message authentication codes (MAC)}

\subsection{ General description of MAC}

We consider two participants, Alice and Bob, connected by a communication line, and Alice sends messages to Bob from time to time.
A third participant, Eve, can spoof (distort) these messages, and the problem Alice and Bob are considering here is the so-called message integrity check: Bob receives message $m$ from Alice and wants to make sure that the message has not been altered by Eve during transmission. This is a message authentication problem, and a common scheme for solving it is as follows:
Alice computes a short message authentication code (tag) $t$ that allows Bob to verify that the message came from Alice (authentication check) and was not tampered with by Eve (integrity check). More precisely, Alice sends Bob a pair $(m,t)$. Upon receiving this pair, Bob checks $t$ according to a certain algorithm and rejects the message if $t$ fails the check. If $t$ passes the check, Bob is sure that the message came from Alice and was not changed during transmission.

We also assume that Alice and Bob share a secret key, which is used by Alice in computing $t$ and in verifying the integrity of the message by Bob.

\begin{definition}\label{th:one} 
 A MAC system    $I=(S,V)$  is a pair of  algorithms $S$   and $V$, where $S$ is called the algorithm for computing $t$ and $V$ is called the algorithm for verifying message integrity, i.e., the algorithm $S$ is used to generate authentication code messages and the algorithm $V$ is used to verify $(m,t)$.
The values $m,k,t$ refer to the sets $M,K,T$ respectively. 
It is 
 assumed that $S$ is a probabilistic algorithm   $t:= S(k,m)$, where $k$ is the authentication key,  $m$ is the message, and $t$ is the authentication code (tag).    
  $V$ is an algorithm that is denoted as $ r : = V(k,m) $, where $r$ is a Boolean variable that takes the values “accept” or “reject”. It is required that the tags generated by $S$  are always accepted by $V$, i.e., MAC
must guarantee that for all keys $k$  and all messages $m$,  $ Pr\{  V(k,m, S(k,m) = "accept" \} = 1$.
Thus, we consider a system in which the verification method is defined as    $   V(k,m,t)  = $  “accept”  if 
$  S(k,m)=t $, and “reject” otherwise.
\end{definition}

\subsection{ MAC security}

The MAC security assessment is based on the following attack game:

\begin{definition}\label{Attack Game} 
Attack Game.
The game involves two participants, called challenger $C$ and adversary $A$.
Let there be MAC system $I = (S; V )$,      $m \in M,  k \in K,  t \in T.$
  The attack game runs as follows:

 1. The challenger picks a random $k$ from $ K$.

2. Adversary $A$ generates message $a$ and sends it to challenger $C$.
$C$ computes $t_0 = S(k,a)$ and sends it to $A$.

 3.   $A$ outputs a candidate forgery pair $(b; t_1^A) $, $b\ne a$,    and sends $b,t_1^A$ to $C$

 $A$ wins this game if    $ V (k; b; t_1^A) =$ accept and 
the advantage of $A$ with respect to $I$ is given  as follows:
${ \it Adv}[A,I] = Pr\{A \, wins\}$. 
\end{definition}
Note that this definition applies to the perfect  MAC, where the secret key $k$ is used only once.


\subsection{ Perfect  MAC}\label{pef}
Here we describe the perfect MAC from \cite{bo}. 
Let  $p$ be a prime number and $Z_p$ be the ring of integers modulo $ p$ (note that $|Z_p|=p$).
We consider the set of messages $M = Z_p^l$, where $l$ is an integer equals  the maximal  length of the  messages, the set of the keys  $K$ is $Z_p^2$, that is any $k$ can be presented as $k = (k_1, k_2)$, where 
$k_1,k_2 \in Z_p$.

Let us consider a message $m = (a_1,a_2, ... , a_\nu),    a_i  \in Z_p, i=1,...,\nu$,  $\nu \le l$. 
The MAC is defined as follows:
\begin{equation}\label{mac} 
S( (k_1,k_2), m) =     ( k_1^{l+1}    + a_1 k_1^{v} + a_2 k_1^{v-1} + ... + a_v  k_1) + k_2       \, ,         
\end{equation}
where $S$   is   the algorithm for computing $t$ and    
 $   V(k,m,t)  = $  “accept”  if 
$  S(k,m)=t $, and “reject” otherwise, see the definition 1  (In (\ref{mac}) all operations over the ring $Z_p$). 
We denote this MAC as $ {\it{\bf P}}$ .

It is proven in \cite{bo} that for all
adversaries $A$  (even inefficient ones) 
\begin{equation}\label{eps} 
 {\it Adv}[A,{\bf P}]  \le (l+1)/ |Z_p| \, .
\end{equation}  
(Recall that  the key $(k_1, k_2)$ 
is uniformly distributed over  $Z_p^2$.)

Here we consider a case where the set of keys and messages are based on the set $Z_p$, but the considered perfect MAC can be based on differen sets.  For example, the other popular  basis is the set of remainders after division by an irreducible polynomial \cite{bo}.

\section{Some properties of the statistical distance }

Statistical distance is one of the most popular measures of deviation of probability distributions \cite{bo}. Its definition is as follows.
\begin{definition}\label{dist} 
Let there be a set $A = \{a_1, ... , n\}, n \ge 1$,  and some probability distributions $P$ and $Q$ on $A$.
Then the statistical distance $\Delta$ is as follows
\begin{equation}\label{dis}
\Delta(P,Q) =  \frac{1} {2} \sum_{i=1}^n | P(a_i) - Q(a_i) | \, .
\end{equation}
\end{definition}

Let $U$ be the uniform distribution, i.e.  $U(a) = 1/n$ for any $a \in A$. 
The next task will be interesting later. Let $P$ be the probability distribution on $A$ and $\Delta(P,U) = \delta$. 
Suppose that $s$ is an integer, $1 \le s < n$, and define 

\begin{equation}\label{pmd}
 P_{max}^s (P,\delta)=  \max_{ \{P: \Delta(P,U) = \delta \} } \,\, \sum_{i=1}^s P(a_i) \,. 
\end{equation}
($ \max  $ here exists, because it is easy to see that    $ \{P: \Delta(P,U) = \delta \} $ is compact set.)
The following statement provides a concise expression for $ P_{max}^s .$
\begin{lemma}\label{sum}  

     \begin{equation}\label{pm}
 P_{max}^s  (P,\delta)=
  \begin{cases}
   \delta + s/n      & \quad \text{if } s\le n(1-\delta)\\
   1  & \quad \text{if }   s > n(1-\delta)\\  \end{cases}
\end{equation}
\end{lemma}
\begin{proof}
Suppose that for some $\delta$ and a distribution $P$
 \begin{equation}\label{maks}
 P_{max}^s (P,\delta) =  \sum_{i=1}^s P(a_i) \,.
\end{equation}

Our goal is to prove (\ref{pm}).  First we define subsets $$A^+ = \{a: P(a) > 1/n \}, \, 
A^- = \{a: 0<P(a) <  1/n \}, 
$$ 
\begin{equation}\label{m}
A^{1/n} = \{a: P(a) = 1/n \} \,
A^0 = \{a: P(a)  = 0\} .
\end{equation}
Taking into account that $\sum_{a \in A} P(a) = 1,$ we can see that 
\begin{equation}\label{del}
 \sum_{a \in A^+} P (a) - |A^+ |/n = \delta
, \, 
\sum_{a \in A^-\cup A^0} P(a)   -   (| A^0| +  |A^-|)/n =  - \delta \, .
\end{equation}
For any  probability distribution $P$   we define three transformations $T_o, T_+$ and $T_-$ as follows:
$T_o(P) $ numbers  letters of $A$ in  such a way that $P(a_1) \ge P(a_2 \ge ... \ge P(a_n)$. 
$T_+(P)$ can be applied if $|A^+| > 1$. If it is so, then, first, $T_o$ is applied  and we obtain 
$P' = T_o(P)$, then we calculate $P^*(a_1) = P'(a_1) + (P'(a_2) - 1/n),  P^*(a_2) = 1/n,  $ 
$P^*(a_3) = P'(a_3), ..., P^*(a_n) = P'(a_n)$,
Note that the set $A^+$ decreases when $T_+$ is applied and $\Delta(P,U) = \Delta(P^*,U)$.

The transformation $T_-(P)$ can be applied if $|A^-| > 1$.   If it is so, take some $a_i, a_j \in A^-$ and let 
$P(i) = 1/n - \tau_1,  P(j) = 1/n - \tau_2,$ $\tau_1 \le \tau_2$.
Then we  define 
$$ P'(k) = P(k)   \quad    \quad \text{if } k \ne i, k \ne j,
$$
$$
P'(i) = 1/n, \, P'(j) = 1/n - (\tau_1+\tau_2)  \quad  \text{if } \,\, \tau_1 + \tau_2 \le 1/n
$$
 \begin{equation}\label{a-}
 P'(i) = 0, \, P'(j) = 2/n  -  (\tau_1+\tau_2)  \quad  \text{if } \,\, \tau_1 + \tau_2  > 1/n
\end{equation}
Then $T_-(P) = T_o(P')$. 
Note that the set $A^-$ decreases when $T_-$ is applied and $\Delta(P,U) = \Delta((T_-(P),U)$.

Now let us define the transformation $T_{ final }$ as follows: first apply the transformation $T_+$ until $|A^+|=1$, and then apply $T_-$ until $|A^-| \le1$ (Recall that the number of applications of $T_+$ and $T_-$ is finite, since after any of these applications the sets $A^+$ or $A^-$ are decreasing.). Let us denote $ T_{final}(P) = P^{final}$.

All transformations $T$ preserve the distance $\Delta$ and hence for any distribution $ \Delta(P^{ final },U) =
\Delta(P,U)  (=\delta )$.
The following properties are true for the distribution $P^{ final }$ by construction:
$|A^+| =1, $  $|A^-| \le1.$
From this, (\ref{del}) and $ \Delta(P^{ final },U) =\delta $ we can see that
$$
 P^{ final }(a_1) = \delta +1/n,  P^{ final }(a_2) =  ... = P^{ final }(a_{k-1} )= 1/n,
$$ $$
P^{ final }(a_k) = 1 - (\delta n - \lfloor \delta n \rfloor )/n ,
P^{ final }(a_{k+1} ) =  ...  = P^{ final }(a_{n} ) = 0, \,  k= n - \lfloor \delta n \rfloor \, .  
$$
The direct calculation of $\sum_{i=1}^s P^{ final }(a_{i})$ gives 
$$
    \sum_{i=1}^s P^{ final }(a_{i})  \,  =
  \begin{cases}
   \delta + s/n      & \quad \text{if } s\le n(1-\delta)\\
   1  & \quad \text{if }   s > n(1-\delta) . \\ \,\end{cases}
$$
As we have shown, the transformations 
 $T_-, T_+, T_o$ do not change distance. From this, 
the latest equation  and (\ref{maks}) we obtain  (\ref{pm}).
\end{proof}

\section{The effect of small deviations from uniform key distribution}
In the second part, we described a perfect  MAC 
 based on polynomial (\ref{mac}).
Here we consider the case when the secret keys $k_1, k_2$ obey distributions slightly different from the uniform distribution.
  
The following theorem is the main result of the paper.
\begin{theorem} \label{th}  
Let  the  MAC$\,\,   \bf{P}$   from   \ref{pef} 
be applied together with $k_1$, $k_2$ which are independent and obey such probability distributions $P_{k_1}$ and $P_{k_2}$ that $\Delta(P_{k_1},U) = \delta_1$, $\Delta(P_{k_2},U) = \delta_2$ for some
non-negative $\delta_1, \delta_2$.
Then
 \begin{equation}\label{eps2} 
 {\it Adv}[A,{\bf P}]  \le  
|Z_p|   P_{\max}^{l+1}(P_{k_1},\delta_1) P^1_{\max}(P_{k_2},\delta_2) \,.
\end{equation}
If $(l+1) \le |Z_p| (1-\delta_1)$ and $1 \le |Z_p| (1-\delta_2)$ then
\begin{equation}\label{eps3}
{\it Adv}[A,{\bf P}]  \le
|Z_p| ( \delta_1 + (l+1)/|Z_p| ) ( \delta_2+ 1/|Z_p|)\,.
\end{equation}  
\end{theorem}
\begin{proof}
For brevity,  for
 $a= (a_1, ... ,  a_\nu )$
let us define
\begin{equation}\label{fg}
 f(a,  k_1) = k_1^{v+1} +
 k_1 ( 
a_1 k_1^{v-1} + a_2 k_1^{v-2} + ... + a_v ) \, , \,\,\,g(a,  k_1, k_2) = f(a, k_1) +k_2.
\end{equation}
Let us describe the Attack Game from Part II for the MAC  $\,\,\bf{P}$.

Let $C$ choose randomly (according to the distributions 
$P_{k_1} , P_{k_2}$) and independently the keys $k^C_1, k^C_2$, and 
$A$ chooses some $a $ and sends $a$ to $C$.
Then $C$ computes $t_0 = g(a,k^C_1,k^C_2)$ and sends it to $A$.
Then $A$ chooses 
some $b$ and $ t_1^A$.
 Perhaps in doing so, $A$ ``guessed'' some keys $k_1^A, k_2^A$.

The adversary $A$ can have two variants or strategies of the game:

1) $A$ guesses 
 $b$ and $t_1^A$ without considering the presence of $a, t_0$.

2) $A$ considers the presence of $a, t_0$ when finding the pair $b, t_1^A$, i.e., it uses the key constraints imposed by $a, t_0$ when searching for $b$ and $t_1^A$.

Consider  the first option. 
$A$ somehow generates $b$ and $t_1^A$. It is easy to see that they correspond to some keys 
$k_1^A,k_2^A$ (Indeed, let us take any $k_1^A$, compute $f(b,k_1^A)$, and then define 
$k_2^A= t_1^A - f(b,k_1^A)$.  Obviously, $g(b,k_1^A, k_2^A)= t_1^A$, see (\ref{fg}).)
Let's proceed to estimate the probability of winning for $A$.  Define  $g(b, k_1^C,k_2^C) = t_1$. Then 
              $$ 
\forall b \in Z_p, \,\,
\forall k_1^A, k_2^A \, \text{for which}  \,\,g(b, k_1^A,k_2^A) = t_1^A \ : \,
$$ 
$$ 
Pr\{A \, \text{ wins } \, \} 
  = Pr\{t_1 = t_1^A \} = Pr\{ g(b, k_1^C,k_2^C) = t_1^A \}= 
$$ 
 \begin{equation}\label{forall}
Pr\{
\sum_{\gamma \in Z_p} f(b, k_1^C) = t_1^A - \gamma \, \, \& \,\,\, k_2^C = \gamma \}.
\end{equation}
Here and below we can assume that formally $Pr\{.\} = P_{k_1} \,P_{k_2}.$ 

Given that the events $ k_2^C = \gamma$ are incompatible at different $\gamma$ and $k_1^C, k_2^C$ are in
dependent, we obtain the following two equalities:  
$$     Pr\{
\sum_{\gamma \in Z_p} f(b, k_1^C) = t_1^A - \gamma  \, \, \& \,\, k_2^C = \gamma \} = \,\sum_{\gamma \in Z_p} Pr\{ f(b, k_1^C) = t_1^A - \gamma  \, \, \& \,\, k_2^C = \gamma \}   $$
 \begin{equation}\label{indep}
= \,\sum_{\gamma \in Z_p}\,    \,P_{k_1} \{ f(b, k_1^C) = t_1^A - \gamma  ) \} \,  P_{k_2} \{k_2^C = \gamma \}
\end{equation}
The latter amount can be represented as follows:
$$
       \sum_{\gamma \in Z_p}\,    \,P_{k_1} \{ f(b, k_1^C) = t_1 ^A- \gamma  ) \} \,  P_{k_2} \{k_2^C = \gamma \}
   =  $$
$$
\sum_{\gamma \in Z_p} P_{k_1} \{  k_1^C  \text{ polynomial root of}\,
f(b,k_1) - t_1^A + \gamma = 0   \}
  P_{k_2}\,  \{k_2^C = \gamma \} = 
$$ 
 \begin{equation}\label{zav} 
\sum_{\gamma \in Z_p} P_{k_1}   (     \sum_{ w \in \{ \text{ polynomial roots} 
f(b,k_1) - t_1^A + \gamma = 0 \}  }    f(b,w) - t_1^A+ \gamma = 0 \,)
 P_{k_2}\,  \{k_2^C = \gamma \}
\end{equation}
Given that the power of the polynomial $ f(b,w) - t_1^A+ \gamma$ does not exceed $l+1$ from 
the definiton of the maxsimum sums   (\ref{pmd}) and Lemma  1 (\ref{pm}), we get an estimate of the last sum:
$$ 
\sum_{\gamma \in Z_p} P_{k_1}  
(     \sum_{ w \in \{ \text{ polynomial roots} 
f(b,k_1) - t_1^A + \gamma = 0 \}  }
f(b,w) - t_1^A + \gamma = 0  )\,\,
 \,  P_{k_2}\,  \{k_2^C = \gamma \} \le 
$$ $$
\sum_{\gamma \in Z_p}   P_{\max}^{l+1}(P_{k_1}, \delta_1)  P^1_{\max}(P_{k_2},\delta_2) = |Z_p|  
  P_{\max}^{l+1}(P_{k_1} \delta_1)  P^1_{\max}(P_{k_2},\delta_2) . 
$$
From this and  the definition 2 of $ { \it Adv}[A,\bf{P]}  $ we obtain
 \begin{equation}\label{1-th}
 { \it Adv}[A,{\bf{P}]  }         \le  |Z_p|    
  P_{\max}^{l+1}(P_{k_1}, \delta_1)  P^1_{\max}(P_{k_2},\delta_2)  \, .
\end{equation}
Therefore, for the first case, the theorem is proved. 

Let us now turn to the analisis of the second case, which is in some ways more natural: $A$ when selecting a pair of $b, t_1^A$ takes into account the presence of $a, t_0$, 
i.e. when searching for $b$ and $t_1^A$ uses $a, t_0$ key limits. 
Use $Z_{a,t_0} $ to denote a set of key pairs $k_1, k_2$ for which $g(a,k_1, k_2) = t_0$. 
Then the probability of winning for $A$ is given by the expression
 $$
\forall b \in Z_p , \,\,
\forall    k_1^A, k_2^A \,  \text{for which}  \,\,g(a, k_1^A,k_2^A) = t_0 \ : \,
$$
$$
Pr\{   A \, \text{ wins  } \, \} = Pr\{t_1 = t_1^A \} =  Pr\{ g(b, k_1^C,k_2^C) = t_1^A \}= 
$$ $$Pr\{
\sum_{\gamma \in Z_{a,t_0}} f(b, k_1^C) = t_1^A - \gamma  \, \, \& \,\, k_2^C = \gamma \}.
$$
Now we can completely repeat the proof from the previous point with the replacement of $Z_p$ with $Z_{a,t_0} $ (starting with (\ref{forall})) and get a similar estimate
\begin{equation}\label{zz}
Pr\{   A \, \text{ wins } \, \}  \le  |Z_{a,t_0}|  
  P_{\max}^{l+1}(P_{k_1}, \delta_1)  P^1_{\max}(P_{k_2},\delta_2)  \, .
\end{equation}
Now let's formally prove a rather obvious fact: $ |Z_{a,t_0}| \le |Z_p|.$ Let's assume the opposite: $ |Z_{a,t_0}| > |Z_p|.$ Then there is (at least one) $k_1$ and a pair $k'_2, k_2''$ are such that $(k_1,k'_2) \in Z_{a,t_0} $ and $(k_1,k''_2) \in Z_{a,t_0} $. By definition $ Z_{a,t_0} $ $f(a,k_1) +k_2' = a$ and $f(a,k_1) + k_2'' = a$, which contradicts the assumption $k_2' \neq k_2''$.

From proven inequality $ |Z_{a,t_0}| \le |Z_p| $ and (\ref{zz}) we get (\ref{eps2}). Thus, (\ref{eps2}) is proved for both cases and the theorem is proved. 
\end{proof}

{\it Comment.} From the proof of the theorem, we see that the probability of ``guessing'' the correct pair $b, t_1$ is the same regardless of whether the $A$ knowledge of the pair takes into account $(a, t_0)$ or ignores it. In a sense, this is a further proof that the described system is ``perfect'' even with small deviations from the uniform distribution.

Let's take a closer look at the assessment from the statement.
Suppose  $\delta_i = \mu_i /|Z_p|$ where $\mu_i \ge 0,  i=1,2$.
Then, from Theorem 1 and  Lemma 1 we can see that 
$ { \it Adv}[A,{\bf{P}]  }    \le \epsilon  $, where $$\epsilon =  
     |Z_p| ( (\mu_1+ (l+1) )/ |Z_p| )  ( (\mu_2+ 1 )/ |Z_p| ) =  (\mu_1+ l+1) (\mu_2 +1) / |Z_p| \, .
  $$
Note that when $\delta_1 = 0, \delta_2 = 0$ (and, accordingly, $\mu_1 = \mu_2 = $0), the latter inequality coincides with the estimate is known for the case without deviations, when the keys follow the uniform distribution. 

\end{document}